\documentclass[11pt, reqno]{jmcs-t}% cls. for J. Math. Comput. Sci.

\usepackage{amsmath, amsthm, amscd, amsfonts, amssymb, graphicx, color}
\usepackage[backref,colorlinks=true]{hyperref}
\usepackage{subfig}
\usepackage{float}
\usepackage{graphics}
\usepackage{adjustbox}
\usepackage{booktabs}
\begin{document}

\title{Constructing Segmented Differentiable Quadratics to Determine Algorithmic Run Times and Model Non-Polynomial Functions}

\author[a1]{Ananth Goyal\corref{c1}}
\ead[c1]{dh.agoyal3@students.srvusd.net}

%use \corref{c1} for corresponding author.
%\ead{Second-author@example.com}

%Use same tags([a1] and [a2],...) for authors and their addresses.  

\address[a1]{Dougherty Valley High School}

%The name will be entered by editor

\cortext[c1]{Corresponding author}
\setcounter{page}{1}
\vol{? (201?)} \pages{1--?}
% The correct information will be entered by the editor

\recivedat{received date}
% The correct dates will be entered by the editor

\authors{A. Goyal}% Enter Short form of authors list here 

\doi{\href{}}
% The correct doi will be entered by the editor

\begin{abstract}
We propose an approach to determine the continual progression of algorithmic efficiency, as an alternative to standard calculations of time complexity, likely, but not exclusively, when dealing with data structures with unknown maximum indexes and with algorithms that are dependent on multiple variables apart from just input size. The proposed method can effectively determine the run time behavior $F$ at any given index $x$ , as well as $\frac{\partial F}{\partial x}$, as a function of only one or multiple arguments, by combining $\frac{n}{2}$ quadratic segments, based upon the principles of Lagrangian Polynomials and their respective secant lines. Although the approach used is designed for analyzing the efficacy of computational algorithms, the proposed method can be used within the pure mathematical field as a novel way to construct non-polynomial functions, such as $\log_2{n}$ or $\frac{n+1}{n-2}$, as a series of segmented differentiable quadratics to model functional behavior and reoccurring natural patterns. After testing, our method had an average accuracy of above of 99\% with regard to functional resemblance.
\begin{keyword}Time Complexity \sep Algorithmic Run Time \sep Polynomials \sep Lagrangian Interpolation
\end{keyword}
\end{abstract}

\maketitle

%Theorem-like structures
%If you need new environments, define them here with the command \newtheorem{envname}{caption}. 

%Some environments such as definitions, theorems, lemmas or examples, have defined  in the list below.

\newtheorem{theorem}{Theorem}[section]
\newtheorem{lemma}[theorem]{Lemma}
\newtheorem{proposition}[theorem]{Proposition}
\newtheorem{corollary}[theorem]{Corollary}
\newtheorem{question}[theorem]{Question}

\theoremstyle{definition}
\newtheorem{definition}[theorem]{Definition}
\newtheorem{algorithm}[theorem]{Algorithm}
\newtheorem{conclusion}[theorem]{Conclusion}
\newtheorem{problem}[theorem]{Problem}

\theoremstyle{remark}
\newtheorem{remark}[theorem]{Remark}
\numberwithin{equation}{section}

\section{Introduction}
Runtime, and it's theoretical subset, time complexity, are imperative to understanding the speed and continual efficiency of all algorithms\cite{Nasar}\cite{Aho}. Particularly because runtime information allows for thorough comparisons between the performance of competing approaches. Due to the varying environments in which algorithms are executed, time complexity is implemented as a function of inputted arguments\cite{Sipser} rather than accounting for the situational execution time\cite{Pusch}; this removes the need to address every extraneous factor that affects the speed of such algorithms\cite{Dean}. There are countless methods on determining the formulaic runtime complexity\cite{Qi}\cite{Guzman}, particularly because, from a theoretical perspective, the true runtime can never be determined without thoroughly examining the algorithm itself\cite{ullman}; however, this does not mean that the process cannot be expedited, simplified, or made easier.

The goal is to produce a function $\mathcal{O}(T(n))$ that can model the time complexity of any given algorithm\cite{Mohr}, primarily who's runtime is defined as a function of more than just a single variable. We define $E(foo(args))$ where $foo(args)$ is any given algorithm and $E$ denotes the execution in a controlled environment. The following method can be used to determine run time with respect to a several variables (not just element size) by evaluating CPU time with respect to an input size. Any confounding variables such as CPU type, computing power, and/or programming language, will be bypassed as they will remain controlled during testing. The constructed polynomial series, which will be a piece-wise of segmented quadratics, will then produce the same functional asymptotic behavior as the true time complexity $\mathcal{O}(T(n))$, which can then be independently determined through the correlation with their respective parent functions. In addition, the methods found for computing such runtimes has profound mathematical implications for representing various non-polynomial functions as differentiable quadratic segments, similar, but not identical, to the outcome of evaluating Taylor Series\cite{Jumarie}\cite{Corliss}. In short, we do this by using reference points of any given non-polynomial, and developing a quadratic (using polynomial interpolation\cite{Boor}) over a particular segment that accurately matches the true functional behavior.

\section{Methods}
Our primary condition is the following:  
$$
\exists_{x\in\mathbb{R}} \left[{F(x+c) - f(x) = \int_{x}^{x+c}\frac{\partial f}{\partial x}}\right]
$$
Additionally, 
$$ \forall(n \in \mathbb{R}: n > 0)\exists\frac{\partial}{\partial n}\mathcal{O}(T(n))$$
This ensures that the targeted Time Complexity function must be constructed of only real numbers and be differentiable throughout, except for segemented bounds. It is important to note that $F(x) \neq \mathcal{O}(T(n)) \vee F(x)\not\approx \mathcal{O}(T(n))$. We also define $E(foo(args)) = k\mathcal{O}(T(n))$, where k is any constant of proportionality that converts the predicted time complexity into execution time or vice-versa. 
\subsection{Lagrangian Polynomial Principles}
We first construct a single line of intersection amongst every consecutive ordered pair of segmented indexes and respective computing time (or any alternative performance modeling metric). We use the following standard point slope formula to do so: 
\begin{equation}
y = \frac{y_i - y_{i-1}}{x_{i} - x_{i-1}}(x-x_{i}) + y_i
\end{equation}

The polynomial of any given segment can be constructed using the explicit formula below\cite{sauer}\cite{Rashed}, in this case the first three indexes within a data set are used; however, this applies for any given 3 point segment within the data set. Defined as: $\forall(x \in (x_j, x_k))|(k = j + 2))$. Note: The proof for the following formulas is shown in section 2.5.
\begin{equation}
\forall(x \in (x_0, x_2)): f(x) =  y_{0}{\frac {(x-x_{1})(x-x_{2})}{(x_{0}-x_{1})(x_{0}-x_{2})}} + y_{1}{\frac {(x-x_{0})(x-x_{2})}{(x_{1}-x_{0})(x_{1}-x_{2})}} + y_{2}{\frac {(x-x_{0})(x-x_{1})}{(x_{2}-x_{0})(x_{2}-x_{1})}}
\end{equation}
We then factor in the polynomial model above and the respective secant line equation, to construct the explicit average form of the initial 3 point segment such that each point is equivalent to the difference between the secant line and the original polynomial. 
\newline
\\
$
\forall(x \in (x_0, x_2)): f(x) =  $
\begin{equation}
y_{0}{\frac {(x-x_{1})(x-x_{2})}{(x_{0}-x_{1})(x_{0}-x_{2})}} + y_{1}{\frac {(x-x_{0})(x-x_{2})}{(x_{1}-x_{0})(x_{1}-x_{2})}} + y_{2}{\frac {(x-x_{0})(x-x_{1})}{(x_{2}-x_{0})(x_{2}-x_{1})}} + 
\left[\frac{y_i - y_{i-1}}{x_{i} - x_{i-1}}(x-x_{i}) + y_i\right]| (i = (1 \vee 2))
\end{equation}
Before we implement this method, we must account for any given segment, and to do so, we must  simplify the method of polynomial construction. First we define F(x) to be dependent on our $f_j$ outputs. 
\begin{equation}
 F(x):=\sum _{j=0}^{k}y_{j}f _{j}(x)\end{equation}
These outputs are determined accordingly (Note: k = 3 in our case; however, the model would work for any value of k): 
\begin{equation}
f _{j}(x):=\prod _{\begin{smallmatrix}0\leq m\leq k\\m\neq j\end{smallmatrix}}{\frac {x-x_{m}}{x_{j}-x_{m}}}={\frac {(x-x_{0})}{(x_{j}-x_{0})}}\cdots {\frac {(x-x_{j-1})}{(x_{j}-x_{j-1})}}{\frac {(x-x_{j+1})}{(x_{j}-x_{j+1})}}\cdots {\frac {(x-x_{k})}{(x_{j}-x_{k})}} \end{equation}

Such that,
\begin{equation}
{\displaystyle \forall ({j\neq i}):f_{j}(x_{i})=\prod _{m\neq j}{\frac {x_{i}-x_{m}}{x_{j}-x_{m}}}=0}
\end{equation}
\subsection{Estimating $\mathcal{O}(T(n))$ as a Function of Quadratic Segments}
We can then average this with the constructed Lagrangian polynomial to get our model for any given 3-point segment. Note: $\because (F(x_k) = F(x_k)) \wedge (\lim_{x \to k^{-}} \frac{\partial F}{\partial x} \neq \lim_{x \to k^{+}} \frac{\partial F}{\partial x}) \therefore \nexists (\frac{\partial F}{\partial x}|_{x = k})$ We can simplify the given expression to\cite{Berrut}: 
\begin{equation}
\forall(x \in (x_j, x_k)): F(x) =  \frac{1}{2}\sum _{j}^{k=j+3}y_{j}\prod _{\begin{smallmatrix}0\leq m\leq k\\m\neq j\end{smallmatrix}}{\frac {x-x_{m}}{x_{j}-x_{m}}} + \frac{1}{2}\left[ \frac{y_{k} - y_{k-1}}{x_{k} - x_{k-1}}(x-x_{k}) + y_k \right]
\end{equation}
Such that,
\begin{equation}
\left.\frac{\partial}{\partial x}\left[\sum _{j}^{k=j+2}y_{j}\prod _{\begin{smallmatrix}0\leq m\leq k\\m\neq j\end{smallmatrix}}{\frac {x-x_{m}}{x_{j}-x_{m}}} + \left[ \frac{y_{k} - y_{k-1}}{x_{k} - x_{k-1}}(x-x_{k}) + y_k \right]\right]\right|_{x = x_{j+1}}\approx \left.\frac{\partial}{\partial x} (\frac{1}{k})E(foo(x_{j+1}))\right|_{x = x_{j+1.5}}\end{equation}
as well as the segmented average\cite{Comenetz}, 
\begin{equation}
\frac{1}{2x_{k} - 2x_{j+1}}\int_{x_j}^{x_k}\left[\sum _{j}^{k=j+2}y_{j}\prod _{\begin{smallmatrix}0\leq m\leq k\\m\neq j\end{smallmatrix}}{\frac {x-x_{m}}{x_{j}-x_{m}}} + \left[ \frac{y_{k} - y_{k-1}}{x_{k} - x_{k-1}}(x-x_{k}) + y_k \right]\right] \approx \int_{x_{j+1}}^{x_k}(\frac{1}{k})E(foo(n))\end{equation}
We then implement the proposed method of each selected segment to construct the function for every iteration of natural numbers by redefining $F(x)$ from a single constructed polynomial to a multi-layered, piece-wise construction of the primary segments of such polynomials. 
\begin{equation}
\forall(x \in \mathbb{R} : x > 0)):  F(x) = \begin{cases} 
       \frac{1}{2}\sum _{0}^{2}y_{j}\prod _{\begin{smallmatrix}0\leq m\leq 2\\m\neq 0\end{smallmatrix}}{\frac {x-x_{m}}{x_{0}-x_{m}}} + \frac{1}{2}\left[\frac{y_2 - y_{1}}{x_{2} - x_{1}}(x-x_{2}) + y_2\right] & x_1 \leq x \leq x_2 \\
      \frac{1}{2}\sum _{1}^{3}y_{j}\prod _{\begin{smallmatrix}2\leq m\leq 4\\m\neq 2\end{smallmatrix}}{\frac {x-x_{m}}{x_{1}-x_{m}}} + \frac{1}{2}\left[\frac{y_3 - y_{2}}{x_{3} - x_{2}}(x-x_{3}) + y_4\right] & x_2 \leq x \leq x_3 \\
        \cdots &   \cdots  \\
    \frac{1}{2}\sum _{n-2}^{n}y_{j}\prod _{\begin{smallmatrix}n-2\leq m\leq n\\m\neq n-2\end{smallmatrix}}{\frac {x-x_{m}}{x_{n-2}-x_{m}}} + \frac{1}{2}\left[\frac{y_n - y_{n-1}}{x_{n} - x_{n-1}}(x-x_{n}) + y_n\right] & x_{n-1} \leq x \leq x_n \\
   \end{cases}
\end{equation}
In order to retrieve the complexity of the algorithm at a particular index $i$ we can now simply compute $F(i)$. Note: $ \nexists\left.\frac{\partial F}{\partial x}\right|_{x = {x_{j}}\vee {x_{k}}}$ but $\forall(x \in (x_j, x_k)\exists\left.\frac{\partial F}{\partial x}\right|_{x}$. Additionally, however, the proposed method, when graphed, will construct a continuous function, making it easy to determine the true runtime of the function as $\mathcal{O}T(n)). $ 
\subsection{Evaluating Quadratic Segments of Multi-variable Algorithms}
\subsubsection{Run Times with Non-Composite Operations}
The following method will suffice if, and only if, the arguments are not directly correlated through any mathematical operation, excluding addition, subtraction, or any non-composite operation. For example, if our unknown time complexity of Algorithm $foo(x, b)$ was $\mathcal{O}(\log_2(x) + b)$. We must first evaluate the execution time with respect to a single variable. We use $E(foo())$, to denote the execution time of the given function; this can be determined by implementing a computing timer into the algorithm.  In this case we evaluate the algorithm accordingly:
\begin{equation}
Y_0 = E(foo(x, b))|\{(x \in \mathbb{N}: x > 0)\wedge(b = 0)\}
\end{equation}

Such that,
\begin{equation}
Y_0 = {y_{0_{0}}\vee{}E(foo(x_0, 0)), y_{0_{1}}\vee{}E(foo(x_1, 0)), \cdots, y_{0_n}\vee{}E(foo(x_n, 0))}\end{equation}
And the same for the other argument: 
\begin{equation}
X_0 = E(foo(x, b))|\{(b \in \mathbb{N}: b > 0)\wedge(x = 0)\}
\end{equation}
Such that,
\begin{equation}
X_0 = {\chi_{0_{0}}\vee{}E(foo(0, b_0)), \chi_{0_{1}}\vee{}E(foo(0, b_0)), \cdots, \chi_{0_n}\vee{}E(foo(0, b_n))}\end{equation}
In this particular case, we first isolate the $F(x, b)$ in terms of x. To do so we must first ensure $x$ and $b$ are independent of each other. Since, in our sample scenario,
\begin{equation}
E(foo(x, b)) = \log_2(x) + b
\end{equation}
Now, we can conclude that, 
\begin{equation}
E(foo(x, 0)) = \log_2(x) + 0 = \log_2(x)
\end{equation}
And, 
\begin{equation}
E(foo(0\vee(\forall \in\mathbb{R} > 0), b)) = \log_2(0\vee(\forall \in\mathbb{R} > 0)) + b
\end{equation}
Now, we can evaluate the $E(foo(x, b))$ over a set of fixed data points. First with respect to x: 
\begin{equation}
F(x, 0)\vee F_x(x, b) =  \frac{1}{2}\sum _{j}^{k=j+3}y_{0_{j}}\prod _{\begin{smallmatrix}0\leq m\leq k\\m\neq j\end{smallmatrix}}{\frac {x-x_{m}}{x_{j}-x_{m}}} + \frac{1}{2}\left[ \frac{y_{0_{k}} - y_{0_{k-1}}}{x_{k} - x_{k-1}}(x-x_{k}) + y_{0_{k}} \right]
\end{equation}
Then with respect to b:
\begin{equation}
F(0, b)\vee F_b(x, b) =  \frac{1}{2}\sum _{j}^{k=j+3}\chi_{0_{j}}\prod _{\begin{smallmatrix}0\leq m\leq k\\m\neq j\end{smallmatrix}}{\frac {b-b_{m}}{b_{j}-b_{m}}} + \frac{1}{2}\left[ \frac{\chi_{0_{k}} - \chi_{0_{k-1}}}{b_{k} - b_{k-1}}(b-b_{k}) + \chi_{0_{k}} \right]
\end{equation}

Once we have computed our segmented quadratics with respect a particular index group, we can construct our piece-wise function of $E(foo(x,b))= \log_2(x) + b$ as two independent, graphical representations. 
\begin{equation}
F(x,b) = \begin{cases}
\forall(x > 0)):F_x= \begin{cases} 
       \frac{1}{2}\sum _{0}^{2}y_{j}\prod _{\begin{smallmatrix}0\leq m\leq 2\\m\neq 0\end{smallmatrix}}{\frac {x-x_{m}}{x_{0}-x_{m}}} + \frac{1}{2}\left[\frac{y_2 - y_{1}}{x_{2} - x_{1}}(x-x_{2}) + y_2\right] & x_1 \leq x \leq x_2 \\
        \cdots &   \cdots  \\
    \frac{1}{2}\sum _{n-2}^{n}y_{j}\prod _{\begin{smallmatrix}n-2\leq m\leq n\\m\neq n-2\end{smallmatrix}}{\frac {x-x_{m}}{x_{n-2}-x_{m}}} + \frac{1}{2}\left[\frac{y_n - y_{n-1}}{x_{n} - x_{n-1}}(x-x_{n}) + y_n\right] & x_{n-1} \leq x \leq x_n \\
   \end{cases}\\
\forall(b > 0)):F_b= \begin{cases}
       \frac{1}{2}\sum _{0}^{2}\chi_{j}\prod _{\begin{smallmatrix}0\leq m\leq 2\\m\neq 0\end{smallmatrix}}{\frac {x-x_{m}}{b_{0}-b_{m}}} + \frac{1}{2}\left[\frac{\chi_2 - \chi_{1}}{b_{2} - b_{1}}(b-b_{2}) + \chi_2\right] & b_0 \leq x \leq b_2 \\
        \cdots &   \cdots  \\
    \frac{1}{2}\sum _{n-2}^{n}\chi_{j}\prod _{\begin{smallmatrix}n-2\leq m\leq n\\m\neq n-2\end{smallmatrix}}{\frac {b-b_{m}}{b_{n-2}-b_{m}}} + \frac{1}{2}\left[\frac{\chi_n - \chi_{n-1}}{b_{n} - b_{n-1}}(b-b_{n}) + y_n\right] & b_{n-1} \leq b \leq b_n \\
   \end{cases}
 \end{cases}
\end{equation}
Although our method produces non-differentiable points at segmented bounds, we can still compute partial derivatives at points $\forall(x\in \mathbb{R}: x > 0)$ such as: 
\begin{equation}
(\left.\frac{\partial}{\partial x})\frac{1}{2}\sum _{j}^{k=j+3}y_{0_{j}}\prod _{\begin{smallmatrix}0\leq m\leq k\\m\neq j\end{smallmatrix}}{\frac {x-x_{m}}{x_{j}-x_{m}}} + \frac{1}{2}\left[ \frac{y_{0_{k}} - y_{0_{k-1}}}{x_{k} - x_{k-1}}(x-x_{k}) + y_{0_{k}} \right]\right|_{x = x_j + 1}\approx (\left.\frac{\partial}{\partial x})(\log_2x + b)\right|_{x = x_j + 1} 
\end{equation}
We can justify the accuracy of $F_x$ with $E(foo(x))$ assuming only one inputted argument accordingly: 
\begin{equation}
\because \nexists (\frac{\partial F}{\partial x}|_{x = j}) \wedge \nexists (\frac{\partial F}{\partial x}|_{x = k}) \wedge \because (\lim_{x \to k^{-}} F_x(x) = \lim_{x \to k^{+}} F_x(x) = F_x(x))\end{equation}
We can conclude that:
\begin{equation}
\frac{1}{n}\sum_{\forall(x \in \mathbb{N}: x > x_{k})}^{n}E(foo(x)) \approx\frac{1}{x_{2} - x_{1}}\int_{x_{1}}^{x_{2}}F_x(x) + \frac{1}{x_{3} - x_{2}}\int_{x_{2}}^{x_{3}}F_x(x) + \cdots + \frac{1}{x_{n} - x_{n-1}}\int_{x_{n-1}}^{x_{n}}F_x(x) 
\end{equation}
Alternatively,  
\begin{equation}
\frac{1}{n}\sum_{\forall(x \in \mathbb{N}: x > x_{k})}^{n}E(foo(x)) \approx\sum_{\forall(x \in \mathbb{N}: x > x_{k})}^{n} \frac{1}{x_{k} - x_{j}}\int_{x_{j}}^{x_{k}}F_x(x)
\end{equation}
\subsubsection{Run times with Composite Operations}
In order to explain the approach used in cases with unknown runtime functions that consist of composite operations, we must implement the following proof.
\begin{theorem}
If the unknown run time function consists of composite operations, such as in $M(x,b) = \frac{1}{b}\log_2(x)$, this can be instantly determined if the functional difference across a set of input values is not just a graphical translation. 
\end{theorem}
\begin{proof}[Proof of Theorem 2.1]
If, 
\begin{equation}G(x, b) = \log_2(x) + b \wedge M(x, b) = \frac{1}{b}\log_2(x)\end{equation}
Then, 
\begin{equation}G(x, 0) = \log_2(x) + 0 = \log_2(x) = G_x(x, b) \vee G(x)\end{equation}
Additionally, 
\begin{equation}G(x, 0) = G_x(x, b)\vee G(x)\end{equation}
But, 
\begin{equation}M(x, 0) \neq M_x(x, b)\vee M(x)\end{equation}
Due to the non-composite operations of $G$, the value of $b$ does not directly impact the value of $x$, rather just the output of the mulitvariable function. The same can be done conversely with other variables; however, if they are directly correlated, such as in $M(x,b)$ it prevents the difference from being just a translation.  
\begin{equation}M(x, 0\vee(\forall(b \in \mathbb{R}: b > 0))) \neq \log_2(x)\end{equation}
Above, it is clear that both independent variables cannot be determined through inputting a constant of 0, causing a non-linear intervariable relationship. 
\end{proof}
In order to construct the primary segmented function with equivalent behavior to the multivariable runtime $E(foo(x,b))$ or with any number of arguments, we must run execution tests with respect to each variable such that the remaining are treated as constants. If, like in the example stated earlier, the unknown runtime function was $\frac{1}{b}\log_2x$, then when graphically modeled for $n$ number of tests in terms of $x$, a set of skewed logarithmic curves would be constructed, where as with respect to $y$, a set of hyperbolic functions would be produced. By treating each temporary, non-functional, constant argument as $k_n$ the graphical differences can be factored in, when creating the single time complexity formula.
Although there are several $k$ values that can be used, to keep the methodology consistent, we decide to take the input value that produces the average functional value over a given segment as the select constant. Although the true function is unknown, we can use the constructed, segmented quadratic to do so. 
\begin{equation}
    \frac{1}{\delta_k - \delta_j}\int_{\delta_j}^{\delta_k}F_\delta(\delta_x, k_b, \cdots, k_z)\partial \delta = a\delta^2 + b\delta + c
\end{equation}
Then we can simply solve for the value of $\delta$ accordingly\cite{Irving}, such that $(\delta\in\mathbb{R})\wedge(\delta > 0)$:
\begin{equation}
   \delta = \frac{-b\pm\sqrt{b^2 - 4a(c - \frac{1}{\delta_k - \delta_j}\int_{\delta_j}^{\delta_k}F_\delta(\delta_x, k_b, \cdots, k_z)\partial \delta)}}{2a}
\end{equation}
While this process is still comprehensible, once we start introducing functions with more than 2 arguments, we must test their values in planes with multiple dimensional layers, rather than just one or two dimensions. To do so, we determine the intervariable relationships between every potential pair of arguments and construct a potential runtime formula accordingly. Note: The higher dimensional order of the function, the more convoluted the formulaic determination becomes.

Suppose the intervariable runtime function $E(foo(x,b,c, \cdots, z))$ such that the corresponding segmented quadratic function is $F(x,b,c, \cdots, z)$. We would evaluate the unknown $E(foo(x,b,\cdots, z))$ with respect to a single variable such that the remaining are treated as constants. Using the example above, we would first plug in constant values into $x$ and $b$, while graphically modeling the rate of change of $c$ as an independent function. Then we begin to adjust $b$ with increments of $i$ to determine, their respective transformational relationship. We would repeat this process for every potential pair $(x,b)$, $(x,c)$, $(b,c)$, and so forth.
\begin{equation}
    F_{x, b}(x, b, \cdots, z) = F_{x}(x, \forall(b\in\mathbb{R}: b > 0: b=b+i), \cdots, k_z)
\end{equation}
\begin{equation}
    F_{b, c}(x, b, \cdots, z) = F_{b}(k_x, b,\forall(c\in\mathbb{R}: c > 0: c=c+i), \cdots, k_z)
\end{equation}
$$
\cdots
$$
\begin{equation}
    F_{c, z}(x, b, \cdots, z) = F_{b}(k_x, k_b, c, \forall(z\in\mathbb{R}: z > 0: z=z+i), \cdots, k_z)
\end{equation}
From their we can use the graphical model to help deduce the formulaic runtime with respect to all variables. 

Similar to the analysis method with respect to a single variable, we can justify the accuracy by approximately equating the average integrated value of each independent segment with it's true, algorithmic counterpart: Note: We define $l$ as the total number of input arguments. 
$$
\frac{1}{(n)(l)}\sum_{\forall(x \in \mathbb{N}: x > x_{k})}^{(n)(l)}E(foo(x, b, \cdots, z)) \approx\frac{1}{x_{2} - x_{1}}\int_{x_{1}}^{x_{2}}F_x((x, b, \cdots, z)\partial x + \cdots + \frac{1}{x_{n} - x_{n-1}}\int_{x_{n-1}}^{x_{n}}F_x(x, b, \cdots, z)\partial x
$$
$$
+ \frac{1}{b_{2} - b_{1}}\int_{b_{1}}^{b_{2}}F_b((x, b, \cdots, z)\partial b + \cdots + \frac{1}{b_{n} - b_{n-1}}\int_{b_{n-1}}^{x_{n}}F_b(x, b, \cdots, z)\partial b 
$$
\begin{equation}
+ \cdots + \frac{1}{z_{2} - z_{1}}\int_{z_{0}}^{z_{1}}F_z((x, b, \cdots, z)\partial z + \cdots + \frac{1}{z_{n} - z_{n-1}}\int_{z_{n-1}}^{x_{n}}F_z(x, b, \cdots, z)\partial b
\end{equation}
This method can be simplified accordingly: 
\begin{equation}
\frac{1}{(n)(l)}\sum_{\forall(x \in \mathbb{N}: x > x_{k})}^{(n)(l)}E(foo(x, b, \cdots, z)) \approx\sum_{x_0}^{z_n}\sum_{\forall(\delta \in \mathbb{R}: x > \delta_{k})}^{n} \frac{1}{\delta_{k} - \delta_{j}}\int_{\delta_{j}}^{\delta_{k}}F(x, b, \cdots, z) \partial \delta
\end{equation}
\subsection{Segmented Quadratics to Construct Non-Polynomials}
The following subsection will discuss the mathematical applications of this method, and will focus on the proofs behind constructing non-polynomials as piece-wise functions built upon segmented quadratics. 
\begin{lemma}
Given $n$ values of $(x \in \mathbb{R})$ with corresponding $n$ values of $(y \in \mathbb{R})$ a representative polynomial $P$ can be constructed such that $\deg(P) < n \wedge P(x_k) = y_k$ 
\end{lemma}

\begin{proof}[Proof of Lemma 2.2]
Let,
\begin{equation}
P_1(x) = \frac{(x - x_2)(x - x_3)\cdots(x - x_n)}{(x_1 - x_2)(x_1 - x_3)\cdots(x_1 - x_n)}
\end{equation}
Therefore, \begin{equation}
P_1(x_1) = 1 \wedge P_1(x_2) = P_1(x_3) = \cdots = P_1(x_n) = 0\end{equation}
Then evaluate, \begin{equation}
P_2, P_3, \cdots, P_n | P_j(x_j) = 1 \wedge P_j(x_i) = 0 \wedge P_j(x_i) = 0 \forall(i \neq j)\end{equation}
Therefore, $P(x) = \sum_{}^{}y_iP_i(x)$ is a constructed polynomial such that $\forall(x_i\in\mathbb{R}: \exists P(x_i)) \wedge \forall(i \in \mathbb{N}: i < n)$. It is built upon subsidiary polynomials of degree $n-1 \therefore \deg(P) < n$ 
\end{proof}
\begin{theorem}
Referencing Lemma 1, given any real non-polynomial, an approximate quadratic piece-wise function $F(x)$ can be constructed using $\frac{n}{2}$ segments produced by 3 values, defined over 2 values, of $x \in \mathbb{R}$ and their corresponding outputs such that $F(x)$ is continuous at all x values including respective transition points, but not necessarily differentiable at such values.
\end{theorem}
%% Example of a proof:
\begin{proof}[Proof of Theorem 2.3]
Since the initial portion of the polynomial is based upon Lemma 1, it is clear that 3 base points will construct a quadratic polynomial, unless their respective derivatives are equivalent which would produce a sloped line. The following method is shown: 
\begin{equation}
    F(x) = \begin{cases} 
       \frac{1}{2}\sum _{0}^{2}y_{j}\prod _{\begin{smallmatrix}0\leq m\leq 2\\m\neq 0\end{smallmatrix}}{\frac {x-x_{m}}{x_{0}-x_{m}}} + \frac{1}{2}\left[\frac{y_2 - y_{1}}{x_{2} - x_{1}}(x-x_{2}) + y_2\right] & x_1 \leq x \leq x_2 \\
      \frac{1}{2}\sum _{1}^{3}y_{j}\prod _{\begin{smallmatrix}1\leq m\leq 3\\m\neq 1\end{smallmatrix}}{\frac {x-x_{m}}{x_{2}-x_{m}}} + \frac{1}{2}\left[\frac{y_3 - y_{2}}{x_{3} - x_{2}}(x-x_{3}) + y_3\right] & x_2 \leq x \leq x_3 \\
        \cdots &   \cdots  \\
    \frac{1}{2}\sum _{n-2}^{n}y_{j}\prod _{\begin{smallmatrix}n-2\leq m\leq n\\m\neq n-2\end{smallmatrix}}{\frac {x-x_{m}}{x_{n-2}-x_{m}}} + \frac{1}{2}\left[\frac{y_n - y_{n-1}}{x_{n} - x_{n-1}}(x-x_{n}) + y_n\right] & x_{n-2} \leq x \leq x_n \\
   \end{cases}
\end{equation}
When simplified, the function would be defined accordingly: 
\begin{equation}
 F(x) = \begin{cases} 
 a_0x^2+b_0x+c_0 & x_1 \leq x \leq x_2 \\
 a_1x^2+b_1x+c_1 & x_2 \leq x \leq x_3 \\
 \cdots & \cdots\\
  a_2x^2+b_2x+c_2 & x_n-1 \leq x \leq x_n \\
 \end{cases}
 \end{equation}
 By definition any polynomial is continuous throughout it's designated bounds\cite{Cucker}, therefore, for all values within each segment, $F(x)$ is continuous. And, since the bounded values of each segment are equivalent, we can conclude that the produced function is continuous everywhere. Formally $(\lim_{x \to t^{-}} F(x) = \lim_{x \to t^{+}} F(x) = F(x))$ where $t$ is any bounded point. However, this does not guarantee that $(\lim_{x \to t^{-}} \frac{\partial F(x)}{\partial x} = \lim_{x \to t^{+}} \frac{\partial F(x)}{\partial x}$; therefore it's derivative at the bounded point is likely undefined. 
\end{proof}
\begin{theorem}
Given any segmented quadratic $F(x) = ax^2 + bx + c$ constructed through averaging Lagrangian Interpolation with it's respective secant, the graphical concavity can be determined by looking at the sign of variable $a$ such that $(a \in \mathbb{R}) \wedge (b \in \mathbb{R}) \wedge (c \in \mathbb{R})$. 
\end{theorem}
\begin{proof}[Proof of Theorem 2.4]
Since $F(x)$ constructed with three base points, the only polynomial function produced are segmented quadratics. Upward concavity exists $\forall(x \in (x_j, x_k)|\frac{\partial^2}{\partial x^2} > 0$;  while downward concavity exists $\forall(x \in (x_j, x_k)|\frac{\partial^2}{\partial x^2} < 0$. However, this process can be expedited without the need to compute second derivatives. 
\newline 
\newline
Since our segmented polynomial is constructed using only three points we can conclude that, 
\begin{equation}
\forall(x \in (x_j, x_k)): \{F_x(x) = ax^2 + bx + c \} | \{(a \in \mathbb{R} : a > 0) \wedge (b \in \mathbb{R} : b > 0) \wedge (c \in \mathbb{R} : c > 0)\}
\end{equation}
Therefore, 
\begin{equation}
\left.\frac{\partial^2 F}{\partial x^2}\right|_{\forall(x \in (x_j, x_k))} = 2a \therefore \iint{2a}\partial x\partial x = F_x(x) \therefore \int_{x_j}^{x_k}\frac{\partial^2 F}{\partial x^2} = \frac{a}{|a|}\left|\int_{x_j}^{x_k}\frac{\partial^2 F}{\partial x^2}\right|
\end{equation}
Thus, the sign of $a$ is the only significant value to determine the segmented concavity $F_x(x)$. Determining the functional concavity of the Lagrangian construction is important, as with certain functions, primarily those with upward concavity, it may not be necessary to compute secant line averages.  
\end{proof}
When testing the mathematical accuracy $A$ of our approach, we use the segmented average value and compare it to that of the original function $G(x)$. In cases where $\int_{a}^{b}G(x) > \int_{a}^{b}F(x)$ we simply compute the reciprocal of the following function. We use $a$ and $b$ as placeholder variables to represent the segmented bounds.
\begin{equation}
     A = \frac{\sum_{0}^n\int_{a}^{b}G(x)}{\sum_{0}^{n}\int_{a}^{b}F(x)}
\end{equation}

\section{Results}
We divide the results section into the primary perspective (algorithmic implications) and the secondary perspective (pure mathematical implications).  
\subsection{Algorithmic Test Results}
We tested our method on four algorithms (two single variable functions and two multivariable functions) and compared the produced time complexity formulas to the true, known, complexities to see how accurate our formulations were, and the extremity of deviations, if, at all. The first algorithm (single variable) was a binary search of $x$ elements, with a complexity of $\mathcal{O}(\log x)$. The second  (single variable) was a sort of $n$ elements, with a complexity of $\mathcal{O}(x\log x)$. The third  (multivariable) was a combined search sort algorithm of $n$ unsorted elements, and $b$ select sorted elements, with a complexity of $\mathcal{O}(b + \log x)$. The fourth  (multivariable) was a custom algorithm of $n$ elements, with a complexity of $\mathcal{O}(mx + \log\log x)$. Although coefficients and additional constants are implemented in the predicted complexity, due to time being the only output variable, the only relevant component is the contents of the big-$O$ as it represents the asymptotic behavior of the algorithmic runtime regardless of any confounding variables. For multivariable algorithms, like stated in the methods, runtime complexities were computed with respect to each variable, and put together in the form of the final predicted complexity. 
\begin{table}[H]
\caption{Predicted Runtime Functions}
\centering
\begin{adjustbox}{width=\columnwidth,center}
\begin{tabular}{cccc}
\toprule
\textbf{Complexity}	& \textbf{Constructed Polynomial}	& \textbf{Predicted Complexity}\\
\midrule
 $\mathcal{O}(\log{x})$	&	 $ F(x) = \begin{cases} 
 -\frac{17}{2400000}x^2+\frac{203}{120000}x+\frac{1947}{6400} & 45 \leq x \leq 95 \\
-\frac{3}{2200000}x^2+\frac{149}{220000}x+\frac{6133}{17600} & 95 \leq x \leq 205 \\
 -\frac{151}{574200000}x^2+\frac{15289}{57420000}x+\frac{270373}{696000} & 205 \leq x \leq 385 \\
 \end{cases}$ & $\frac{1}{11}\mathcal{O}(\log{x}) + 0.22$\\\\
$\mathcal{O}(x\log{x})$ &	 $ F(x) = \begin{cases} 
 \frac{19}{625000}x^2+\frac{4}{3125}x+\frac{69}{500} & 50 \leq x \leq 100 \\
  \frac{3}{312500}x^2+\frac{123}{25000}x-\frac{9}{500} & 100 \leq x \leq 125 \\
 \frac{19}{625000}x^2+\frac{4}{3125}x+\frac{69}{500} & 125 \leq x \leq 150 \\
 \end{cases}$ & $\frac{1}{350}\mathcal{O}(x\log{}x)$\\\\
 $\mathcal{O}(mx\log{x})$ &	$\begin{cases} F_x(x, m) = \begin{cases} 
 \frac{19}{625000}x^2+\frac{4}{3125}x+\frac{69}{500} & 50 \leq x \leq 100 \\
  \frac{3}{312500}x^2+\frac{123}{25000}x-\frac{9}{500} & 100 \leq x \leq 125 \\
 \frac{19}{625000}x^2+\frac{4}{3125}x+\frac{69}{500} & 125 \leq x \leq 150 \\
 \end{cases}\\ F_m(x, m) = \begin{cases} 
 1892m & 0 \leq m \leq 10 \\
 1837m - 33461 & 10 \leq m \leq 20 \\
 2066m - 56184 & 20 \leq m \leq 30 \\
 \end{cases} \end{cases}$	& $\frac{1159}{3750}\mathcal{O}(mx\log{x})$	\\\\
 $\mathcal{O}(mx+\log\log{b})$			& $\begin{cases} F_x(x, m, b) = \begin{cases} 
 \frac{19}{625000}x^2+\frac{4}{3125}x+\frac{69}{500} & 50 \leq x \leq 100 \\
  \frac{3}{312500}x^2+\frac{123}{25000}x-\frac{9}{500} & 100 \leq x \leq 125 \\
 \frac{19}{625000}x^2+\frac{4}{3125}x+\frac{69}{500} & 125 \leq x \leq 150 \\
 \end{cases} \\ F_m(x, m, b) = \begin{cases} 
 1892m & 0 \leq m \leq 10 \\
 1837m - 33461 & 10 \leq m \leq 20 \\
 2066m - 56184 & 20 \leq m \leq 30 \\
 \end{cases}\\F_b(x, m, b) = \begin{cases} 
-\frac{13}{12000000}x^2 + \frac{863}{1200000}x + \frac{138949}{160000} & 5 \leq x \leq 305 \\
-\frac{3}{8000000}x^2 + \frac{363}{800000}x + \frac{283677}{320000} & 305 \leq x \leq 505 \\
  -\frac{1}{15000000}x^2 + \frac{51}{250000}x + \frac{560389}{600000} & 505 \leq x \leq 705 \\
 \end{cases} \end{cases}$	& $\mathcal{O}(mx\log{}x+\log\log{b}) + 0.06$\\
\bottomrule
\end{tabular}
\end{adjustbox}
\end{table}
\subsection{Mathematical Test Results}
We tested our method against various non-polynomial functions, and selected one example of each common type of non-polynomial to be representative of the accuracy with it's functional family. We made sure to remove any non-composite components such as any constants, as that would only adjust the function by a graphical translation. The functions used were $\log_2x$ (logarithm family), $\frac{x-1}{x}$ (rational family), $2^x$ (exponential family), and $4cos(x)$ (trigonometric family); although we could choose more convoluted functions, we wanted to showcase performance for functions that are similar to their parent functions to attain a holistic perspective. In most cases, the Lagrangian constructions tend to exceed the vertical level of their respective non-polynomials, making secant line averages most useful with downward concavity. We defined our piece wise function until some relatively close, arbitrary whole value that leaves numbers simple, to stay consistent; however, the accuracy is still indicative of potential performance regardless of the final bound, due to the natural progression of such functions. 
\begin{table}[H]
\caption{Accuracy of Constructed Polynomials Compared to Base Function}
\centering
%% \tablesize{} %% You can specify the fontsize here, e.g., \tablesize{\footnotesize}. If commented out \small will be used.
\begin{tabular}{ccc}
\toprule
\textbf{Function}	& \textbf{Constructed Polynomial} & \textbf{Calculated Accuracy}\\
\midrule
 $\log_2x$	&	 $ F(x) = \begin{cases} 
 -\frac{1}{196}x^2+\frac{1}{4}x+\frac{4}{3} & 8 \leq x \leq 16 \\
-\frac{1}{768}x^2+\frac{1}{8}x+\frac{7}{3} & 16 \leq x \leq 32 \\
 -\frac{1}{3072}x^2+\frac{1}{16}x+\frac{10}{3} & 32 \leq x \leq 64 \\
 \end{cases}$		& $99.964\%$\\\\
$\cos(\pi{}x)$ &	 $ F(x) = \begin{cases} 
 -\frac{414}{125}x^2-\frac{43}{125}x+1 & 0 \leq x \leq 0.5 \\
 \frac{414}{125}x^2-\frac{871}{125}x+ \frac{332}{125} & 0.5 \leq x \leq 1 \\
\frac{414}{125}x^2-\frac{157}{25}x+ \frac{246}{125} & 1 \leq x \leq 1.5 \\
 \end{cases}$	& $99.79\%$\\\\
$2^x$ &	$ F(x) = \begin{cases} 
 2x^2-6x+8 & 3 \leq x \leq 4 \\
 4x^2-20x+32 & 4 \leq x \leq 5 \\
  8x^2-56x+112 & 5 \leq x \leq 6 \\
 \end{cases}$			& $98.92\%$\\\\
$\frac{x-1}{x}$			& $F(x) = \begin{cases} 
 -\frac{1}{16}x^2+\frac{1}{2}x+\frac{1}{4} & 2 \leq x \leq 4 \\
-\frac{1}{128}x^2+\frac{1}{8}x+\frac{3}{8} & 4 \leq x \leq 8 \\
  -\frac{1}{1024}x^2+\frac{1}{32}x+\frac{11}{16} & 8 \leq x \leq 16 \\
 \end{cases}$		&	$99.34\%$\\
\bottomrule
\end{tabular}
\end{table}

\begin{figure}[H]
  \centering
  \subfloat[Without Partitions]{\includegraphics[width=0.32\textwidth]{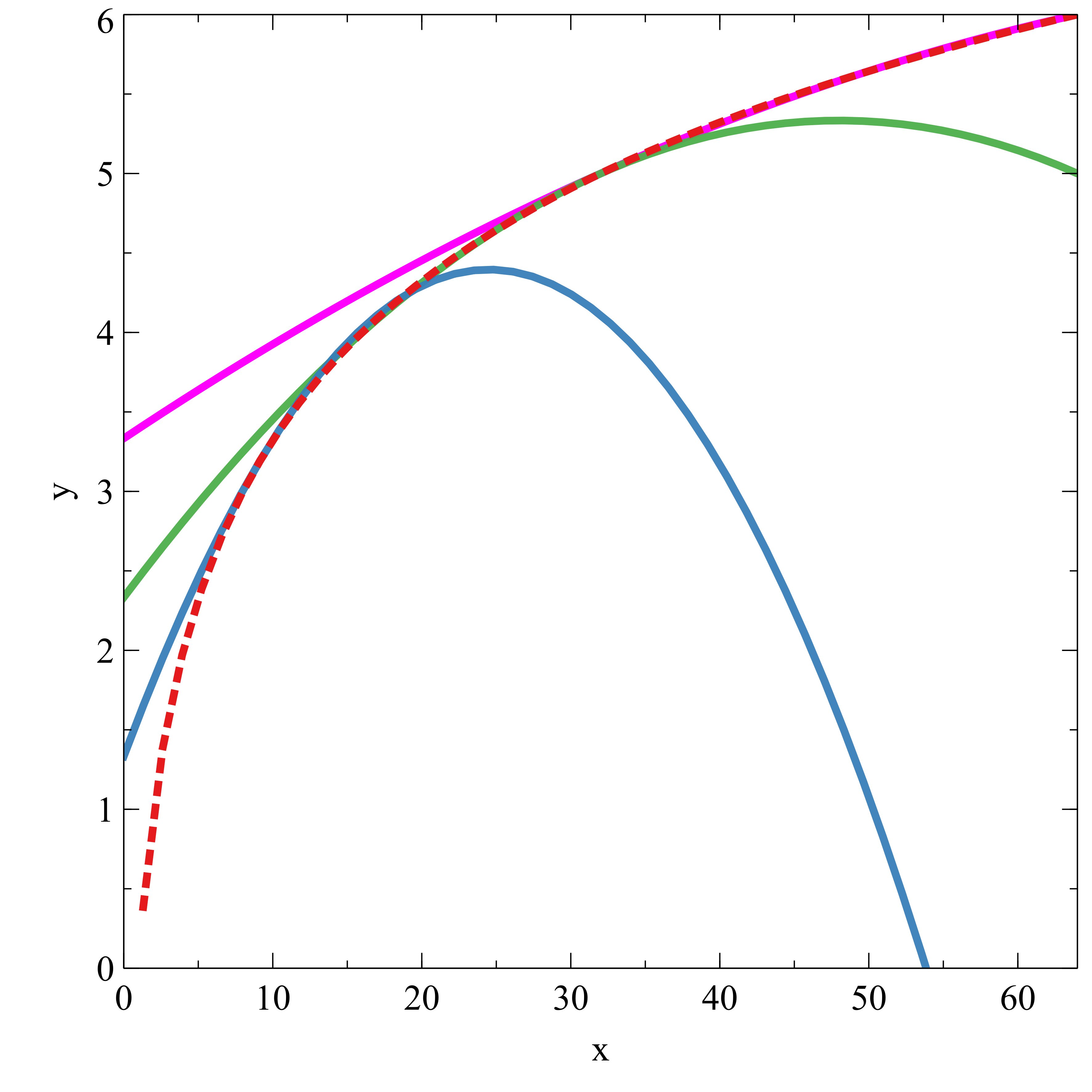}\label{fig:f1}}
  \hfill
  \centering
  \subfloat[With Partitions]{\includegraphics[width=0.32\textwidth]{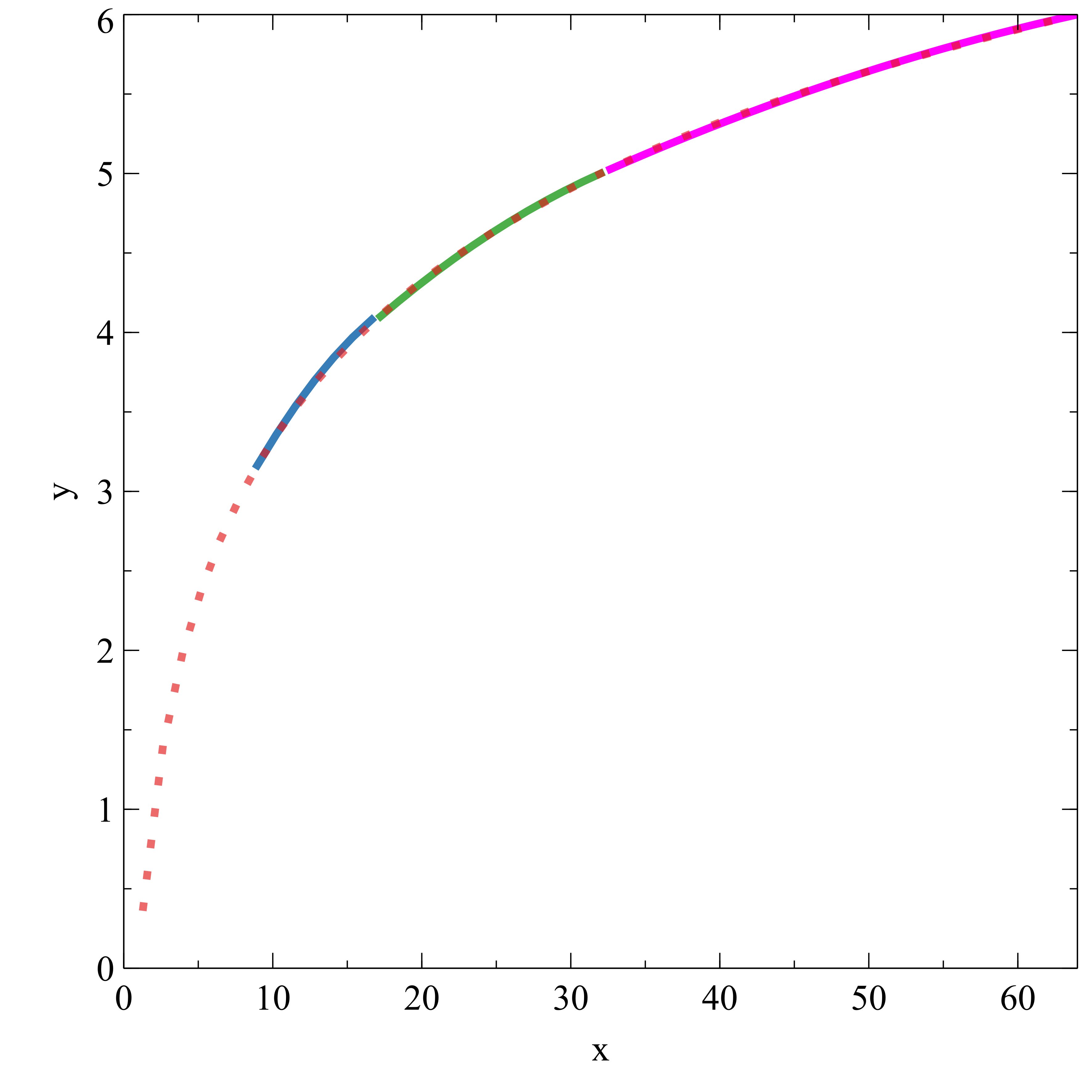}\label{fig:f2}}
  \caption{Graphical Representations of $\log_2x$ and $F(x)$}
\end{figure}

\begin{figure}[H]
  \centering
  \subfloat[Without Partitions]{\includegraphics[width=0.32\textwidth]{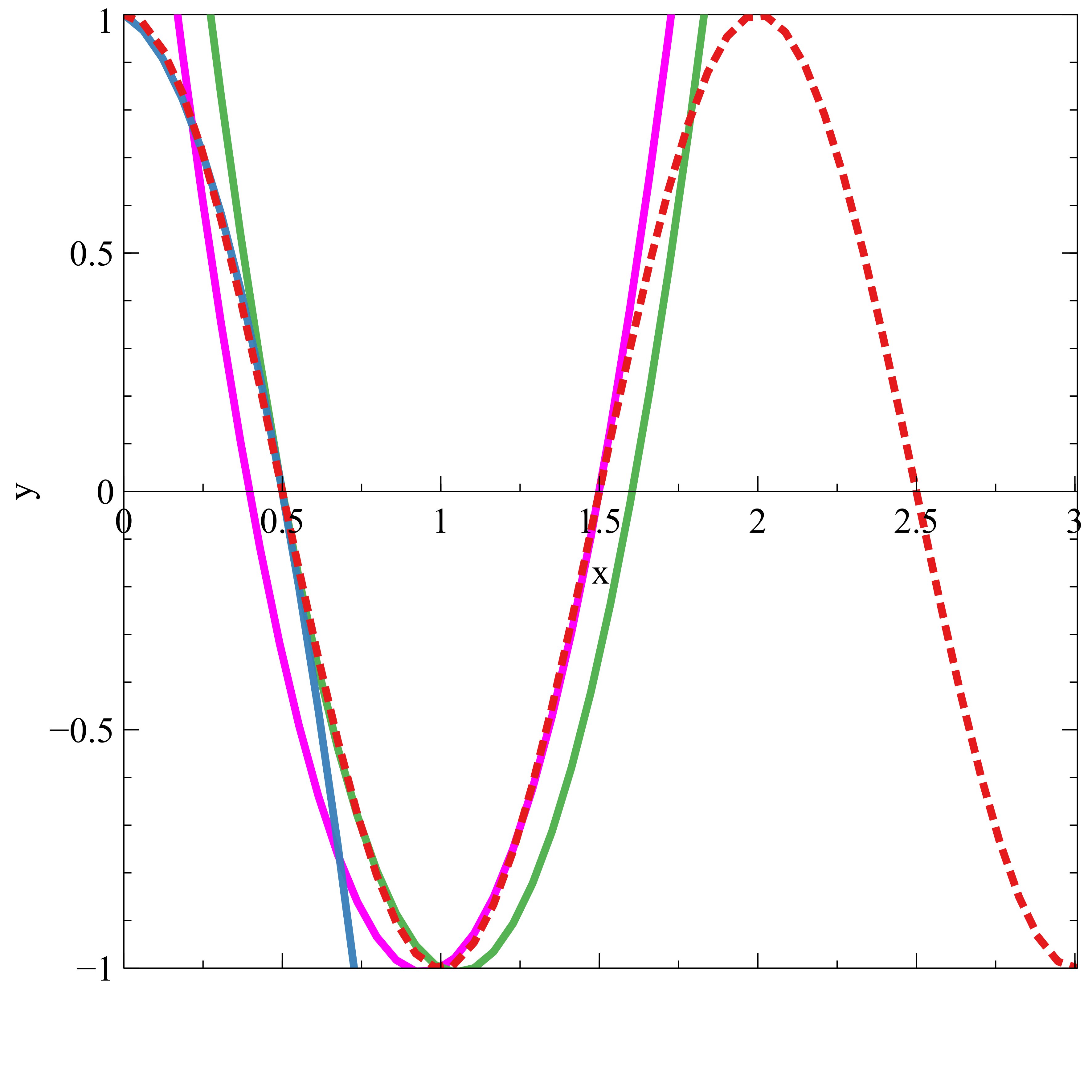}\label{fig:f1}}
  \hfill
  \centering
  \subfloat[With Partitions]{\includegraphics[width=0.32\textwidth]{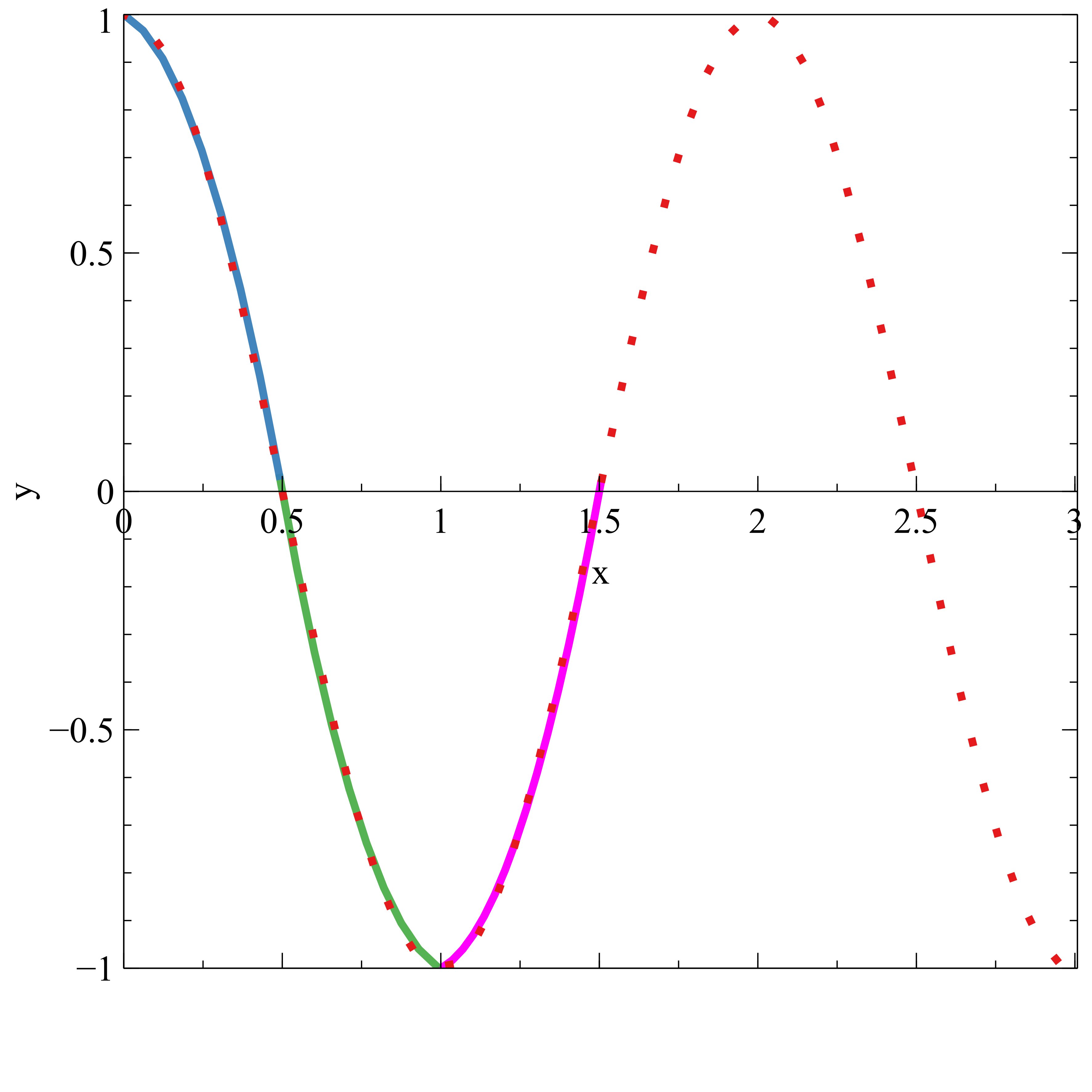}\label{fig:f2}}
  \caption{Graphical Representations of $\cos(\pi{}x)$ and $F(x)$}
\end{figure}
\begin{figure}[H]
  \centering
  \subfloat[Without Partitions]{\includegraphics[width=0.32\textwidth]{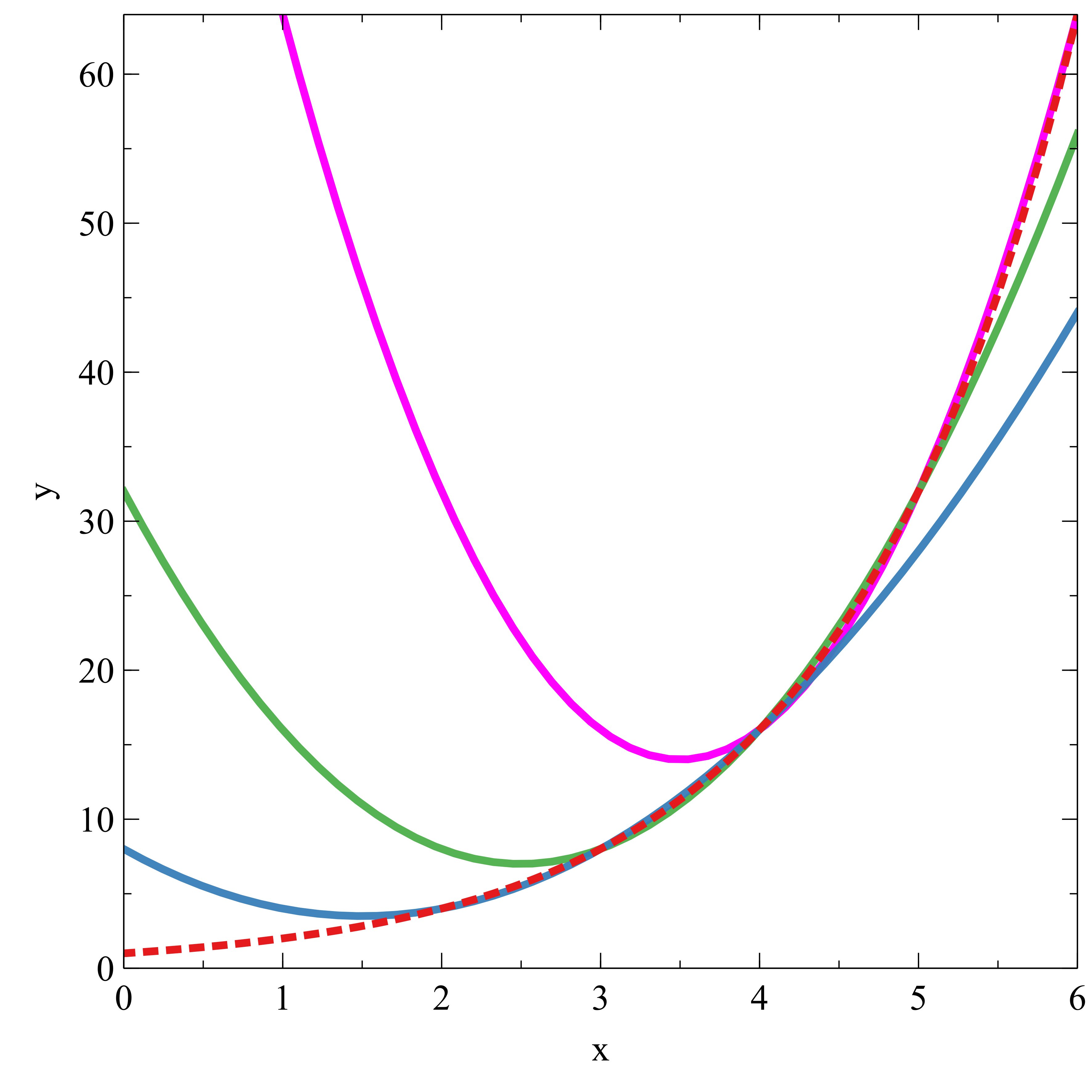}\label{fig:f1}}
  \hfill
  \centering
  \subfloat[With Partitions]{\includegraphics[width=0.32\textwidth]{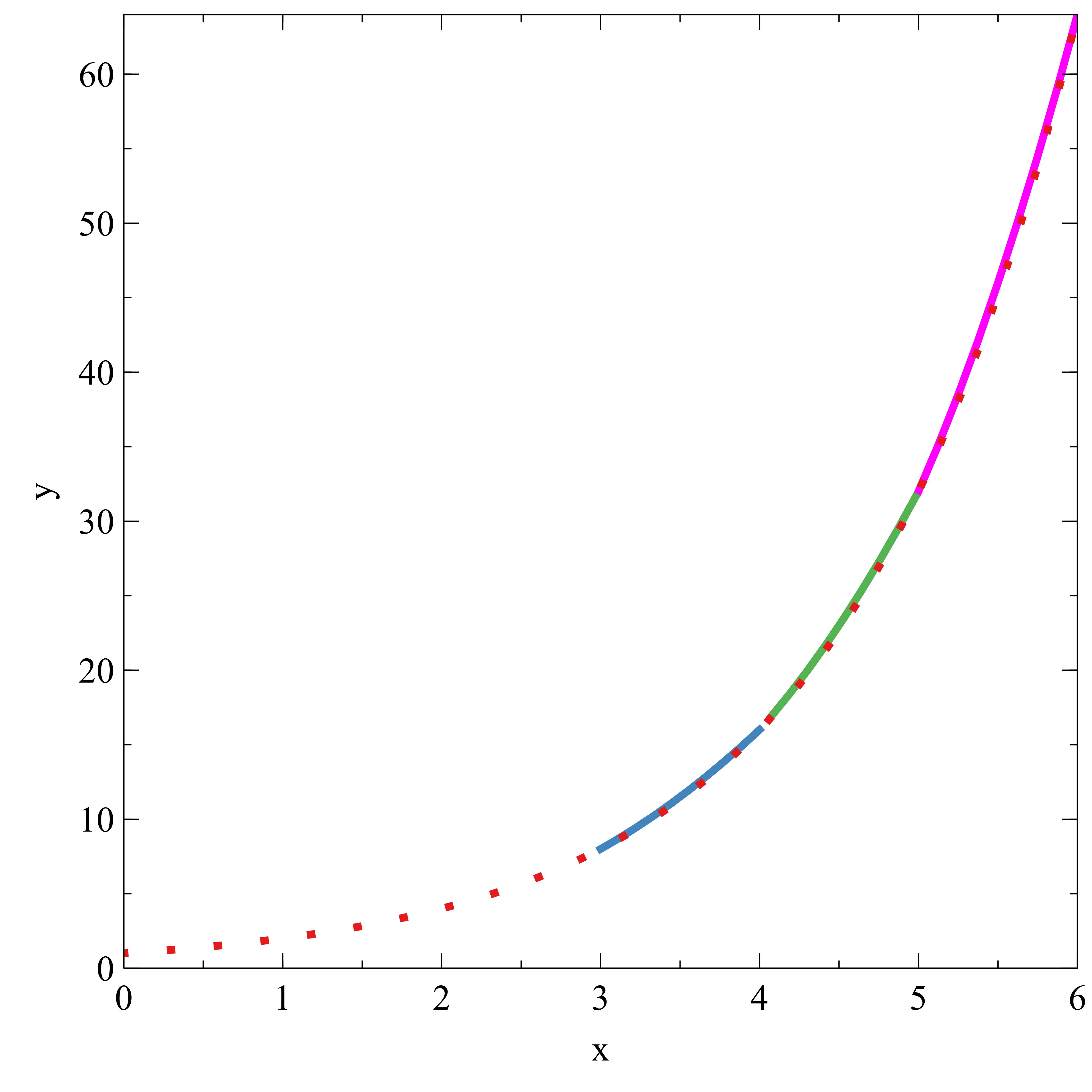}\label{fig:f2}}
  \caption{Graphical Representations of $2^x$ and $F(x)$}
\end{figure}
\begin{figure}[H]
  \centering
  \subfloat[Without Partitions]{\includegraphics[width=0.32\textwidth]{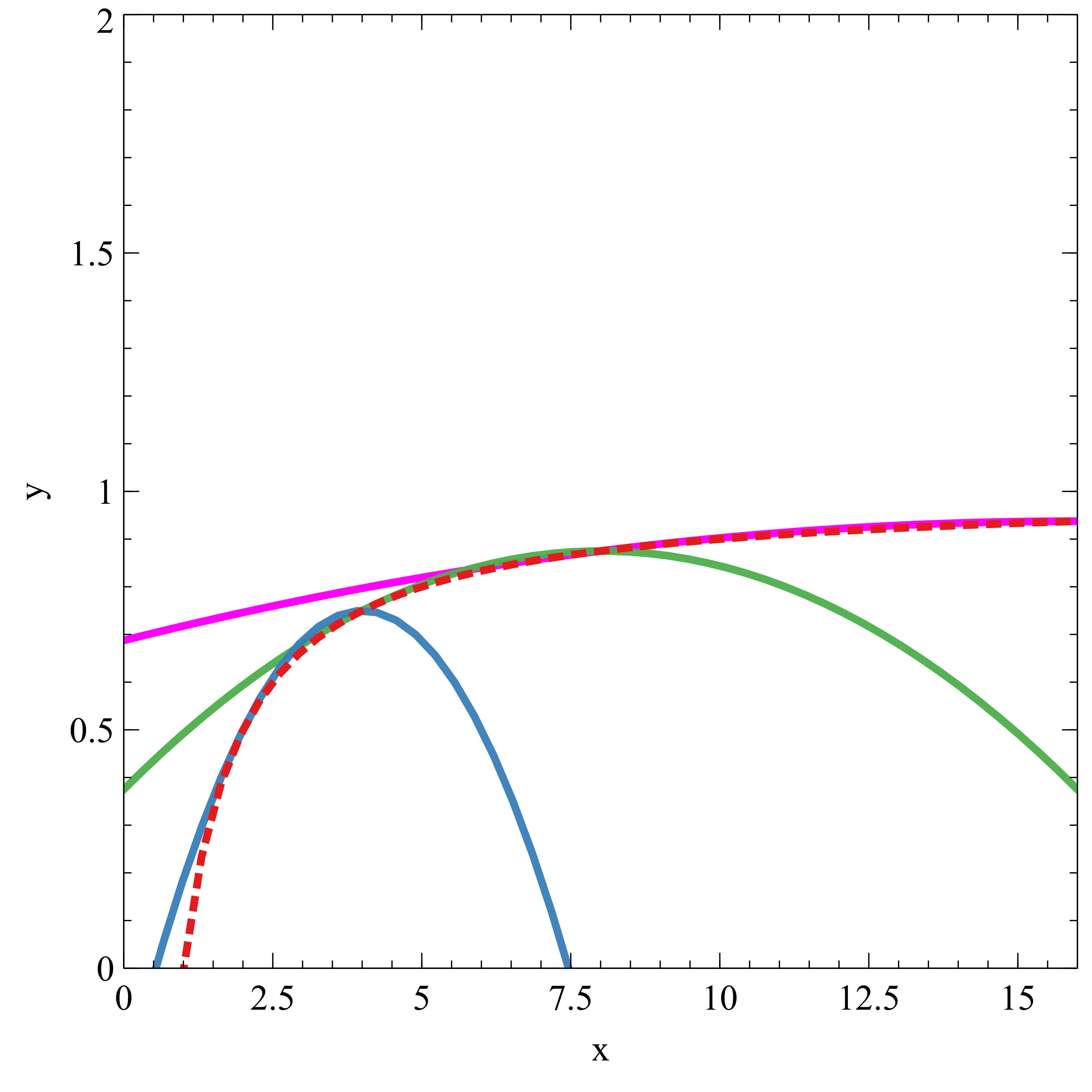}\label{fig:f1}}
  \hfill
  \centering
  \subfloat[With Partitions]{\includegraphics[width=0.32\textwidth]{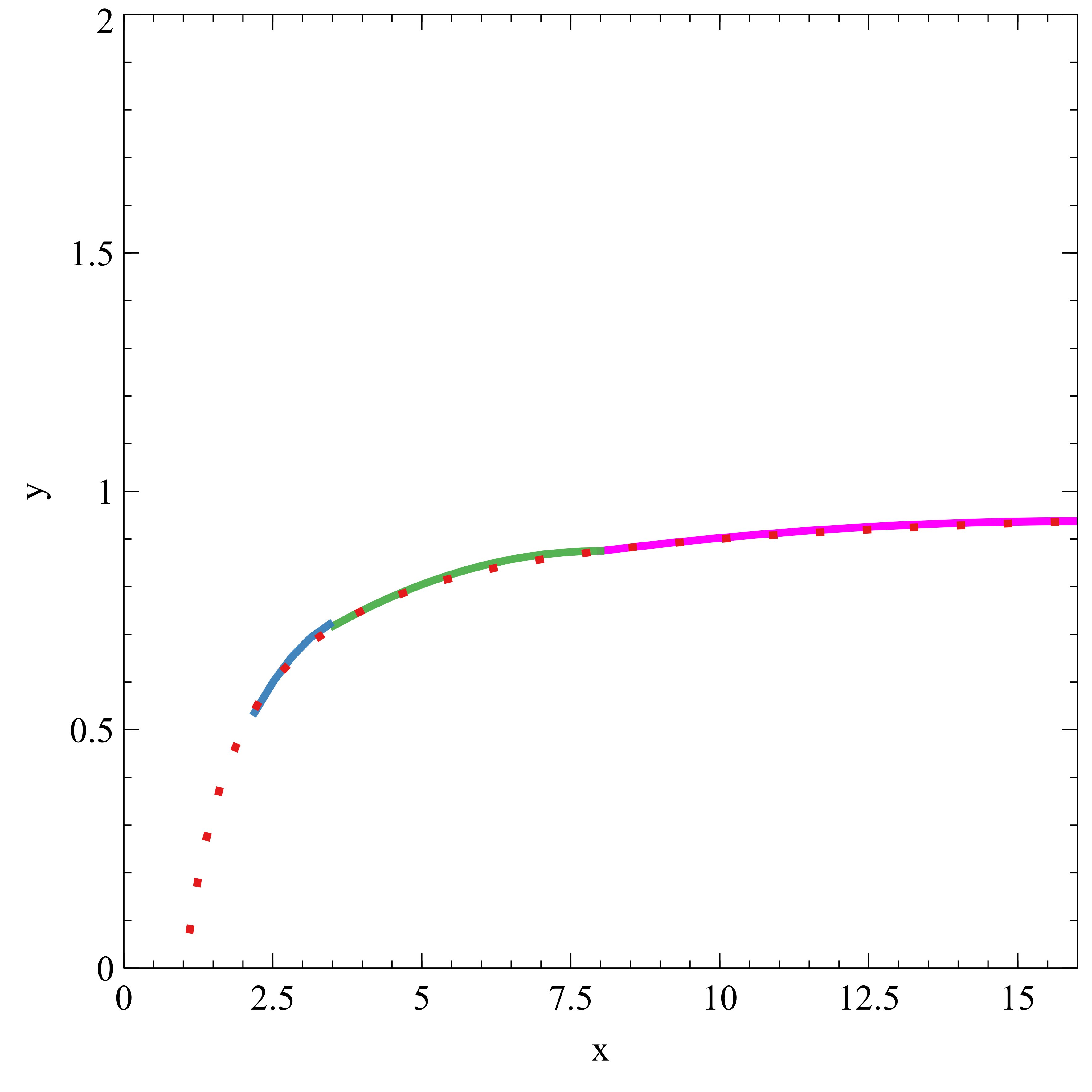}\label{fig:f2}}
  \caption{Graphical Representations of $\frac{x-1}{x}$ and $F(x)$}
\end{figure}

\section{Discussion and Implications}

\subsection{Algorithmic Discussion}
After testing the proposed approach against several known algorithms, we were able to swiftly determine the runtime functions that correspond with their true time complexities. We tested the approach on two single variable algorithms and two multivariable algorithms such that we could compare the produced complexity behaviour with the true, known, complexity. In practice, this method will be used on algorithm's where complexities are unknown to help determine their runtime functions; however, experimentally, we needed to know the true complexity beforehand to deduce the comparative accuracy. Regardless, in all cases our method was able to produce the correct big-$O$ runtime function. This was determined through the automated construction of segmented polynomial models given a set of input data. By treating each variable independently and graphing their grouped correlation, it made it easy to deduce the respective time complexity. To reiterate, any external coefficients and constants are a result of the particular test environment and because time is the output value. The only relevant component in determining the accuracy of the method are the contents of the big-O function. Most of the predicted runtime complexity functions followed the format of $k\mathcal{O}(T(n)) + C$ where $k$ is the constant of proportionality between execution time and standard time complexity and $C$ is any factor of translation that matches the produced graphical curve/line with their true counterpart. While these values help us overlay our constructions with their parent functions, they aren't necessarily important in determining the accuracy of our approach as the asymptotic behavior of our construction will be the same regardless. We are confident that the proposed method can significantly help expedite the process of determining functional time complexities in all cases, including both single and mulitvariable algorithms. 
\subsection{Mathematical Discussion}
After reviewing the results, we were able to confirm the accuracy of the proposed approach with constructing matching segmented differentiable quadratics given any non-polynomials. These include logarithmic, exponential, trigonometric, and rational functions. To determine the approaches accuracy with select functions, we calculated the average value of the formulated function over a particular segment. And after doing so, as well as reviewing the formulaic relationships between computed segments, we found a collective functional resemblance score of greater than $99\%$ and began to notice profound mathematical implications. After testing just a few data points, we can produce a rule that can construct the next consecutive segmented polynomial based upon the functional patterns that surface. For example, with regard to $\log_2x$, we were able to determine that every consecutive segment was equivalent to $\frac{a}{4}x^2 + \frac{b}{2}x + (c+1)$ such that the variables are the constants of the previous segment and the accuracy of any additional segments would remain identical. Not only is this a revolutionary method for accurate, polynomial replications; but it's sheer simplicity combats flaws found in leading methods of doing so (primarily with non-sinusoidal functions), most notably Taylor series. Using these methods mathematicians and scientists can construct accurate, differentiable functions to represent patterned data, non-polynomials, and functions found in higher theoretical dimensions. Additionally, a similar approach can be used to determine the natural progression of repetitious systems such as natural disasters, planetary orbits, or pandemic-related death tolls, to lead to a better understanding of their nature. As in theory, their physical attributes and properties are built upon reoccurring, natural functions. 
\subsection{Conclusion}
In this paper we proposed an approach to use segmented quadratic construction, based upon the principles of Lagrangian Interpolation to help determine algorithmic runtimes, as well as model non-polynomials with advanced, foreseeable applications in pure mathematics and pattern modeling/recognition found in science and nature. We hope to build upon this approach by improving and determine new ways to apply this research in all computational and mathematical based fields. 

\section*{Acknowledgments}
I would like to thank Professor Jeffery Ullman, Mr. Sudhir Kamath, Mr. Robert Gendron, Mr. Phillip Nho, and Ms. Katie MacDougall for their continual support with my research work.

% References should be listed in alphabetical order according to the surnames of the first author at the end of the paper and should be cited in the text as, e.g., [2] or [3,Theorem 4], etc.;

% For paper detail check, please refer to:  http://www.ams.org/mrlookup

\end{document}